\documentclass[11ppt]{article}

\usepackage[mathcal]{euscript}
\usepackage{mathpazo}

\usepackage[top=1.25in,bottom=1.25in,left=1.25in,right=1.25in]{geometry}

\usepackage{enumerate}
\usepackage{subfig}

\usepackage{microtype}
\usepackage{fixltx2e}
\usepackage{fullpage}
\usepackage{xspace}
\usepackage{tikz}  
\usepackage{multirow}
\usepackage{multicol}

\usepackage{graphicx,amssymb,amsmath,amsthm}
\usepackage{hyperref}
\usepackage[utf8]{inputenc}
\usepackage{marvosym}
\usepackage{cite}
\usepackage{verbatim}
\usepackage{color}
\usepackage{stmaryrd}
\usepackage{wrapfig,color}
\usepackage{lineno}

\def\fakeparagraph#1{\par\medskip\noindent\textbf{#1}}

\newcommand{\T}		{{\mathcal{T}}}
\newcommand{\C}		{{\mathcal{C}}}
\newcommand{\X}		{{\mathbb{X}}}
\newcommand{\reals}	{{\mathbb{R}}}
\newcommand{\ms}	{{\mathcal{X}}}
\newcommand{\mt}	{{\mathsf{M}}}
\newcommand{\eps}	{{\varepsilon}}

\newcommand{\shift}	{{\sigma}}
\newcommand{\distortion}	{\mathrm{Dist}}

\newcommand{\etal} {\emph{et al.}}

\definecolor{darkred}{rgb}{0.8, 0.0, 0.0}

\newtheorem{theorem}{Theorem}[section]
\newtheorem{corollary}[theorem]{Corollary}
\newtheorem{lemma}[theorem]{Lemma}
\newtheorem{claim}[theorem]{Claim}
\theoremstyle{definition}

\makeatletter
\long\def\@makecaption#1#2{
   \vskip 10pt
   \setbox\@tempboxa\hbox{{\footnotesize \textbf{#1.} #2}}
   \ifdim \wd\@tempboxa >\hsize         
       {\footnotesize \textbf{#1.} #2\par}
     \else                              
       \hbox to\hsize{\hfil\box\@tempboxa\hfil}
   \fi}
\makeatother

\begin{document}

\title{Computing the Gromov-Hausdorff Distance for Metric Trees\thanks{Work on this paper by P. K.
  Agarwal, K. Fox and A. Nath was supported by NSF under grants CCF-09-40671, CCF-10-12254,
  CCF-11-61359, and IIS-14-08846, and by Grant 2012/229 from the U.S.-Israel
Binational Science Foundation. A. Sidiropoulos was supported by NSF under grants CAREER-1453472 and CCF-1423230. Y. Wang was supported by NSF under grant CCF--1319406.}}

\author{
Pankaj K. Agarwal\\ Duke University\\\texttt{pankaj@cs.duke.edu} \and
Kyle Fox\\ Duke University\\\texttt{kylefox@cs.duke.edu} \and
Abhinandan Nath\\ Duke University\\\texttt{abhinath@cs.duke.edu} \and
Anastasios Sidiropoulos\\ Ohio State University\\\texttt{sidiropoulos.1@osu.edu} \and
Yusu Wang\\ Ohio State University\\\texttt{yusu@cse.ohio-state.edu}
}

\date{}
%
%
%
%


\maketitle

\begin{abstract}
The Gromov-Hausdorff (GH) distance is a natural way to measure distance between two metric spaces. 
We prove that it is $\mathrm{NP}$-hard to approximate the Gromov-Hausdorff distance
better than a factor of~$3$ for geodesic metrics on a pair of trees.
We complement this result by providing a polynomial time $O(\min\{n, \sqrt{rn}\})$-approximation algorithm for computing the GH distance between a pair of metric trees, where~$r$ is
the ratio of the longest edge length in both trees to the shortest edge length. For metric trees with unit length edges, this yields an $O(\sqrt{n})$-approximation algorithm. 
\end{abstract}
\newpage
\section{Introduction}


The Gromov-Hausdorff distance (or GH distance for brevity)~\cite{GH-book} is one of the most natural distance measures between metric spaces, and has been used, for example, for matching deformable shapes~\cite{ms-tcfii-05,bbk-eciid-06}, and for analyzing hierarchical clustering trees \cite{clust-um}.
Informally, the Gromov-Hausdorff distance measures the \emph{additive} distortion suffered when
mapping one metric space to another using a correspondence between their points.
Multiple approaches have been proposed to estimate the Gromov-Hausdorff distance~\cite{ms-tcfii-05,bbk-eciid-06,f-oghds-07}.

Despite much effort, the problem of computing, either exactly or approximately, GH distance has remained elusive. The problem is not known to be $\mathrm{NP}$-hard, and computing the GH distance, even approximately, for \emph{graphic metrics}\footnote{A graphic metric measures the shortest path distance between vertices of a graph with unit length edges.} is at least as hard as the graph isomorphism problem. Indeed, the metrics for two graphs have GH distance 0 if and only if the two graphs are isomorphic. Motivated by this trivial hardness result, it is natural to ask whether GH distance becomes easier in more restrictive settings such as geodesic metrics over trees, where efficient algorithms are known for checking isomorphism~\cite{algo}.

\paragraph{Related work.}
Most work on associating points between two metric spaces involves \emph{embedding} a given high
dimensional metric space into an infinite host space of lower dimensional metric spaces.
However, there is some work on finding a bijection between points in two given finite metric spaces that minimizes typically multiplicative distortion of distances between points and their images, with some limited results on additive distortion.

Kenyon \etal~\cite{krs-ldmbp-09} give an optimal algorithm for minimizing the multiplicative distortion of a bijection between two equal-sized finite metric spaces, and a parameterized polynomial time algorithm that finds the optimal bijection between an arbitrary unweighted graph metric and a bounded-degree tree metric.

Papadimitriou and Safra~\cite{ps-clebp-05} show that it is NP-hard to approximate the multiplicative distortion of any bijection between two finite 3-dimensional point sets to within any additive constant or to a factor better than 3.

Hall and Papadimitriou~\cite{hp-ad-05} discuss the \emph{additive distortion problem} -- given two equal-sized point sets $S,T \subset \reals^d$, find the smallest $\Delta$ such that there exists a bijection $f : S \to T$ such that $d(x,y) - \Delta \le d(f(x),f(y)) \le d(x,y) + \Delta$. They show that it is NP-hard to approximate by a factor better than 3 in $\reals^3$, and also give a 2-approximation for $\reals^1$ and a 5-approximation for the more general problem of embedding an arbitrary metric space onto $\reals^1$. However, there setting differs from ours in two major ways -- firstly, they consider finite metric spaces of equal size, whereas in this paper the metric spaces may be uncountably infinite; secondly, they consider bijections between metric spaces, whereas in our work we deal with correspondences between metric spaces which are more general than bijections. Thus, their approach cannot be easily extended to our setting.

%

The interleaving distance between merge trees~\cite{morozov} was proposed as a measure to
compare functions over topological domains that is stable to small perturbations in a function. 
Distances for the more general Reeb graphs are given in \cite{bauer, DMP14}. These concepts are related to the GH distance (Section \ref{sec:GH-interleaving}), which we will leverage to design an approximation algorithm for the GH distance for metric trees.

\paragraph{Our results.}
In this paper, we give the first non-trivial results on approximating the GH distance between metric trees.
First, we prove (in Section~\ref{section:hardness}) that the problem remains $\mathrm{NP}$-hard even for metric trees via a reduction from \textsc{3-Partition}. In fact, we show that there exists no algorithm with approximation ratio less than~$3$ unless $\mathrm{P}=\mathrm{NP}$. As noted above, we are not aware of any result that shows the GH distance problem being $\mathrm{NP}$-hard even for general graphic metrics.

To complement our hardness result, we give an $O(\sqrt{n})$-approximation algorithm for
the GH distance between metric trees with $n$ nodes and \emph{unit length} edges.
Our algorithm works with arbitrary edge lengths as well; however, the approximation ratio
becomes~$O(\min\{n, \sqrt{rn}\})$ where~$r$ is the ratio of the longest edge length in both trees
to the shortest edge length.
Even achieving the $O(n)$-approximation ratio presented here for arbitrary~$r$ is a non-trivial
task.

Our algorithm uses a reduction, described in Section~\ref{sec:GH-interleaving}, to the similar problem of computing the \emph{interleaving distance}~\cite{morozov} between two \emph{merge trees}.
Given a function~$f : \X \rightarrow \mathbb{R}$ over a topological space~$\X$,
the merge tree~$T_f$ describes the connectivity between components of the sublevel sets of $f$ (see Section~\ref{sec:prelims} for a more formal definition).
Morozov et al.~\cite{morozov} proposed the interleaving distance as a way to
compare merge trees and their associated functions\footnote{In fact, our hardness result can be easily extended to the GH distance between graphic metrics for trees and the interleaving distance between merge trees.}.
We describe, in Section~\ref{sec:compute_id}, an $O(\min\{n,\sqrt{rn}\})$-approximation algorithm for interleaving distance between merge trees, and our reduction provides a similar approximation for computing the GH distance between two metric trees.

\section{Preliminaries}
\label{sec:prelims}
\paragraph{Metric Spaces and the Gromov-Hausdorff Distance.} A \emph{metric space} $\ms = (X, \rho)$ consists of a (potentially infinite) set $X$ and a function $\rho : X \times X \rightarrow \mathbb{R}_{\geq 0}$ such that the following hold: $\rho(x,y) = 0$ iff $x = y$; $\rho(x,y) = \rho(y,x)$; and $\rho(x,z) \leq \rho(x,y) + \rho(y,z)$.

Given sets $A$ and $B$, a \emph{correspondence} between $A$ and $B$ is a set $\C \subseteq A \times B$ such that: (i) for all $a \in A$, there exists $b \in B$ such that $(a,b) \in \C$; and (ii) for all $b \in B$, there exists $a \in A$ such that $(a,b) \in \C$.
We use $\Pi(A,B)$ to denote the set of all correspondences between $A$ and $B$.

Let $\ms_1 = (X_1, \rho_1)$ and $\ms_2 = (X_2, \rho_2)$ be two metric spaces. The \emph{distortion} of a correspondence $\C \in \Pi(X_1,X_2)$ is defined as:
\begin{align*}
\distortion(\C) = \sup_{(x,y), (x', y') \in \C} | \rho_1(x, x') - \rho_2(y, y') | .
\end{align*}

The \emph{Gromov-Hausdorff distance}~\cite{f-oghds-07}, $d_{GH}$, between $\mathcal{X}_1$ and $\mathcal{X}_2$ is defined as:
\begin{align*}
  d_{GH}(\mathcal{X}_1, \mathcal{X}_2) = \frac{1}{2} \inf_{\C \in \Pi(X_1,X_2)} \distortion(\C).
\end{align*}

Intuitively, $d_{GH}$ measures how close can we get to an \emph{isometric} (distance-preserving) embeddding between two metric spaces. We note that there are different equivalent definitions of the Gromov-Hausdorff distance; see e.g, Theorem 7.3.25 of \cite{MG-book} and Remark 1 of \cite{f-oghds-07}.

Given a tree $T = (V,E)$ and a length function $l : E \to \mathbb{R}_{\geq0}$, we associate a metric space $\T = (|T|, d)$ with $T$ as follows. $|T|$ is a geometric realization of $T$. The metric space is extended to points in an edge such that each edge of length $l$ is isometric to the interval $[0,l]$. For $x,y \in |T|$, define $d(x,y)$ to be the length of the path $\pi(x,y) \in |T|$ which is simply the sum of the lengths of the restrictions of this path to edges in $T$. It is clear that $d$ is a metric. The metric space thus obtained is a \emph{metric tree}. We often do not distinguish between $T$ and $|T|$ and write $\T = (T,d)$.

\paragraph{Merge Trees and the Interleaving Distance.}
Let $f : \X \rightarrow \mathbb{R}$ be a continuous function from a connected topological space $\X$ to the set of real numbers.
The \emph{sublevel set} at a value $a \in \mathbb{R}$ is defined as $f_{\leq a} = \{x \in \X \mid f(x) \leq a\}$. A \emph{merge tree} $\mt_f$ captures the evolution of the topology of the sublevel sets as the function value is increased continuously from $- \infty$ to $+ \infty$. 
Formally, it is obtained as follows.
Let $\text{epi } f = \{(x,y) \in \X \times \mathbb{R} \mid y \ge f(x)\}$.
Let $\bar{f} : \text{epi } f \rightarrow \mathbb{R}$ be such that $\bar{f}((x,y)) = y$.
We may say $\bar{f}((x,y))$ is the \emph{height} of point $(x,y) \in \X \times \mathbb{R}$.
For two points $(x,y)$ and $(x',y')$ in $\X\times \mathbb{R}$ with $y = y'$, let $(x,y) \sim (x',y')$ denote them lying in the same component of $\bar{f}^{-1} (y) (= \bar{f}^{-1}(y'))$.
Then $\sim$ is an equivalence relation, and the merge tree $\mt_f$ is defined as the quotient space $(\X \times \mathbb{R}) / \sim$.

Since two components of $\bar{f}^{-1}$ at a certain height can only merge at a higher height and a component can never split as height increases, we get a rooted tree where the internal nodes represent the points where two components merge and the leaves represent the birth of a new component at a local minimum. Figure \ref{merge_tree} shows an example of a merge tree for a 1-dimensional function. Note that the merge tree extends to a height of $\infty$, and our assumption that $\X$ is connected implies we have only one component in $F_{\leq \infty}$.
We define the \emph{root} of merge tree $\mt_f$ to be the internal node with the highest function value (if there are no internal nodes, the only  leaf is defined to be the root).

\begin{figure}[htb]
\centering
\includegraphics[height=0.18\textheight]{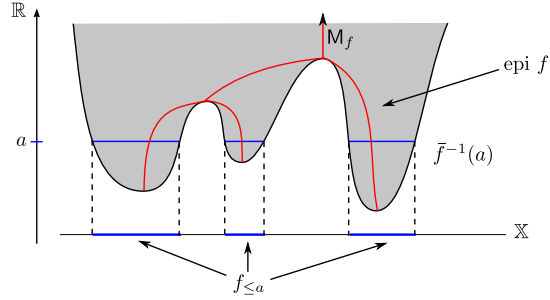}
\caption{\label{merge_tree}Merge tree $\mt_f$ (shown in red) for a function $f : \X \rightarrow \mathbb{R}$, where $\X = \mathbb{R}$. $\text{epi } f$ is shown in grey, and $\bar{f}^{-1}(a)$ is the grey region below the blue horizontal lines in $\text{epi } f$.}
\end{figure}

Since each point $x \in \mt_f$ represents a component of a sublevel set at a certain height, we can
associate this height value with $x$, denoted by $\hat{f}(x)$. Given a merge tree $\mt_f$ and $\eps \geq 0$, an
$\eps$-shift map $\shift_f^\eps : \mt_f \rightarrow \mt_f$ is the map that maps a point $x \in \mt_f$ to its ancestor at height $\hat{f}(x)+\eps$, i.e., $\hat{f}(\shift_f^\eps(x)) =
\hat{f}(x) + \eps$.
Given $\eps \geq 0$ and merge trees $\mt_f$ and $\mt_g$, two continuous maps $\alpha : \mt_f \rightarrow \mt_g$ and $\beta : \mt_g \rightarrow \mt_f$ are said to be $\eps$-compatible if they satisfy the following conditions~:
\begin{equation}
\label{eq:inter_cond}
\begin{aligned}
  \hat{g}(\alpha(x)) &= \hat{f}(x) + \eps, \forall x \in \mt_f;
  &\hat{f}(\beta(y)) &= \hat{g}(y) + \eps , \forall y \in \mt_g;\\
  \beta \circ \alpha &= \shift_f^{2\eps};
  &\alpha \circ \beta &= \shift_g^{2\eps}.
\end{aligned}
\end{equation}

\begin{figure}[htb]
\centering
\includegraphics[width=0.35\textwidth]{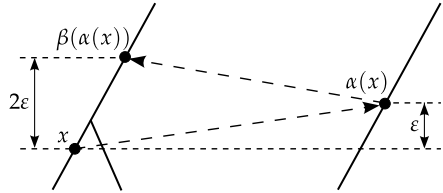}
\caption{\label{fig:interleaving}Part of trees $\mt_f$ and $\mt_g$ showing $\alpha$ and $\beta$.}
\end{figure}
See Figure~\ref{fig:interleaving} for an example.
The \emph{interleaving distance} \cite{morozov} is then defined as
\begin{align*}
d_I(\mt_f, \mt_g) = \inf \{\eps \ge 0 \mid \text{there exist $\eps$-compatible maps $\alpha$ and $\beta$}\}.
\end{align*}


\fakeparagraph{Remark.} We can relax the requirements on~$\alpha$ and~$\beta$ from their normal definitions as follows.

(i) Instead of requiring \emph{exact} value changes, we require 
\begin{align*}
\label{eq:def_mod}
\hat{f}(x) \leq \hat{g}(\alpha(x)) \leq \hat{f}(x) + \eps, \forall x \in \mt_f;\enspace\enspace
\hat{g}(y) \leq \hat{f}(\beta(y)) \leq \hat{g}(y) + \eps, \forall y \in \mt_g.
\end{align*}

(ii) If $x_1$ is an ancestor of $x_2$ in $\mt_f$, then $\alpha(x_1)$ is an ancestor of $\alpha(x_2)$ in $\mt_g$. A similar rule applies for~$\beta$.

(iii) $\beta(\alpha(x))$ must go to an ancestor of~$x$ and $\alpha(\beta(y))$ must go to an ancestor of~$y$.

Any pair of maps satisfying the original requirements also satisfies the relaxed requirements for the same value of $\eps$. Conversely, for any pair of maps satisfying the relaxed requirements, we can \emph{stretch up} the images for each map as necessary so that the new maps satisfy the original requirements, without changing the value of $\eps$. Thus, both definitions of interleaving distance are equivalent. For convenience, when two $\eps$-compatible maps are given to us we assume that they satisfy \eqref{eq:inter_cond}, but we construct $\eps$-compatible maps that satisfy the relaxed conditions mentioned, knowing that they can be ``stretched'' as just described to satisfy \eqref{eq:inter_cond}.

\smallskip
If we know $\alpha(x)$ for a point $x$ at height $h$, then we can compute $\alpha(y)$ for any ancestor $y$ of $x$ at height $h' \geq h$ by simply putting $\alpha(y) = \shift_f^{h'-h}( \alpha(x))$. A similar claim holds for $\beta$. Thus specifying the maps for the leaves of the trees suffices, because any point in the tree is the ancestor of at least one of the leaves. Hence, these maps have a representation that requires linear space in the size of the trees.

As shown in \cite{morozov}, the interleaving distance is a metric and has the desirable properties of being both stable to small function perturbations and more discriminative than the popular bottleneck distance between persistence diagrams \cite{ceh-spd-07}.


\section{Hardness of Approximation}
\label{section:hardness}
We now show the hardness of approximating the GH distance by a reduction from the following decision problem called \emph{balanced partition} (or \textsc{BAL-PART} for brevity): given a multiset of positive integers $X = \{a_1, \ldots, a_n\}$, and an integer $m$ such that $1 \le m \le n$, is it possible to partition $X$ into $m$ multisets $\{X_1, \ldots, X_m\}$ such that all the elements in each multiset sum to the same quantity $\mu=\left(\sum_{i=1}^n a_i\right)/m$? We prove below that \textsc{BAL-PART} is \emph{strongly} NP-complete, i.e., it remains NP-complete even if $a_i \le n^c$ for some constant $c \ge 1$ for all $1 \le i \le n$.
\begin{lemma}
\label{lem:np}
\textsc{BAL-PART} is strongly NP-complete.
\end{lemma}
\begin{proof}
We reduce \textsc{3-PARTITION}, a strongly NP-complete problem \cite{gary} to \textsc{BAL-PART}. Given a multiset of positive integers $Y = \{a_1, \ldots, a_n\}$ with $n=3m$, \textsc{3-PARTITION} asks to partition $Y$ into $m$ multisets $\{Y_1, \ldots, Y_m\}$ of size $3$ each so that the elements in each multiset sum to the same quantity. Given a \textsc{3-PARTITION} instance, we construct an instance of \textsc{BAL-PART} as follows.

Basically, we add a sufficiently large number to each $a_i$ so that if two multisets of the new numbers have the same sum, they have the same number of elements. In particular, let $\bar{a}= \sum_{i=1}^n a_i$. Then set $a_i' = a_i + \bar{a}$ and $X = \{a_1', \ldots, a_n'\}$. This reduction takes polynomial time, and the new numbers are polynomially larger than the original ones. We show that there exists an appropriate partition of $Y$ iff there exists an appropriate partition of $X$.

Suppose there exists an appropriate partition $\{Y_1, \ldots, Y_m\}$ of $Y$. Then setting $X_i = \{a_j' \mid a_j \in Y_i\}$ for $i = 1, \ldots, m$ gives us the desired partition of $X$.

Conversely, suppose there exists an appropriate partition $\{X_1, \ldots, X_m\}$ of $X$. Suppose $|X_i| = n_1 > |X_j| = n_2$ for some $i\neq j$. We thus have
\begin{equation}
\sum_{a_k'\in X_i} a_k + n_1 \bar{a} = \sum_{a_k'\in X_j} a_k + n_2 \bar{a}
\Rightarrow (n_1 - n_2) \bar{a} = \sum_{a_k'\in X_j} a_k - \sum_{a_k'\in X_i} a_k
\Rightarrow \sum_{a_k'\in X_j} a_k - \sum_{a_k'\in X_i} a_k \geq \bar{a},
\end{equation}
a contradiction since $\sum_{a_k'\in X_j} a_k < \bar{a}$. Thus, each partition $X_i$ is of equal size. Since $n=3m$, the size of each $X_i$ is 3.
\end{proof}

\begin{figure}[htb]
\centering
\includegraphics[width=0.55\textwidth]{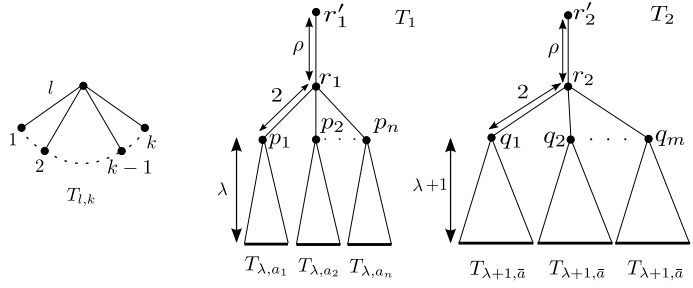}
\caption{\label{fig:reduction}The trees $T_{l,k}$, $T_1$ and $T_2$.}
\end{figure}

We now reduce an instance of \textsc{BAL-PART}, in which each $a_i \le n^c$ for some constant $c \ge 1$, to GH-distance computation.
Given an instance $X = \{a_1, \ldots, a_n\}$ and $1 \le m \le n$ of \textsc{BAL-PART}, we construct two trees $T_1$ and $T_2$ as follows. Let $\lambda > 6$ and $\rho < \lambda - 6$ be two positive constants. Let $T_{l,k}$ denote a star graph having $k$~edges, each of length~$l$. $T_1$ consists of a node $r_1$ incident on an edge $(r_1, r'_1)$ of length~$\rho$ and on $n$ edges $\{(r_1,p_1), \ldots, (r_1,p_n)\}$ of length~2, where $p_i$ is the center of a copy of $T_{\lambda,a_i}$. $T_2$ consists of a node $r_2$ incident on an edge $(r_2,r'_2)$ of length $\rho$ and to $m$ edges $\{(r_2,q_1), \ldots, (r_2,q_m)\}$ of length~2, where each $q_i$ is the center of a distinct copy of $T_{\lambda+1,\bar{a}}$, and $\bar{a}=\left(\sum_{i=1}^n a_i\right)/m$. See Figure~\ref{fig:reduction} for an illustration. 
We refer to the edges of $T_{\lambda, a_i}$ in $T_1$ and copies of $T_{\lambda+1, \bar{a}}$ in $T_2$ as \emph{bottom} edges. Let $\T_1$ and $\T_2$ denote the metric trees associated with $T_1$ and $T_2$ respectively. Since $\lambda, \rho$ are constants and $a_i \le n^c$ for all $1 \le i \le n$, this construction can be done in polynomial time.

\begin{lemma}
\label{lem:reduction}
If $(X,m)$ is a \emph{yes} instance of \textsc{BAL-PART}, then $d_{GH}(\T_1, \T_2) \leq 1$. Otherwise, $d_{GH}(\T_1, \T_2) \geq 3$.
\end{lemma}
\begin{proof}
Suppose $X$ can be partitioned into $m$ subsets $X_1, \ldots, X_m$ of equal weight $\bar{a}=\left(\sum_{i=1}^n a_i\right)/m$. We construct a correspondence $\C$ between $\T_1$ and $\T_2$ with distortion at most 2, implying that $d_{GH}(\T_1, \T_2) \le 1$. A linearly interpolated bijection between the points of edges $(r_1,r'_1)$ and $(r_2,r'_2)$, with $r_1$ mapping to $r_2$ and $r'_1$ mapping to $r'_2$, is added to $\C$. If $a_i \in X_j$,  the linearly interpolated bijection between edges $(r_1,p_i)$ and $(r_2,q_j)$ is added to $\C$. Also, the leaves of $T_{\lambda,a_i}$ are each mapped to a distinct leaf of $T_{\lambda+1,\bar{a}}$ attached to $q_j$ such that there is a bijection between the leaves of $T_1$ and $T_2$ -- this is possible since $T_{\lambda+1,\bar{a}}$ has $\bar{a}$ leaves, and $\sum_{a \in X_j} a = \bar{a}$. The interior points of these leaf edges are mapped using linear interpolation. Overall, the distortion induced by $\C$ is at most $2$ -- this stems from the fact that $\C$ is piecewise linear, and the difference between the length of any path in one tree and its image under $\C$ in the other tree is at most $2$.

Suppose $d_{GH}(\T_1, \T_2) < 3$, and let $\C$ be a correspondence between $\T_1$ and $\T_2$ with distortion $<6$. Consider two leaves $l,m \ne r'_2$ in $\T_2$. Then $d(l,m) \ge 2 \lambda + 2$. Let $l',m'$ be their corresponding images in $\T_1$ under $\C$. We argue that $l', m'$ lie on distinct bottom edges of $\T_1$. Indeed, since $\distortion(\C) < 6$, the distance between $l'$ and $m'$ is $d(l',m') > d(l,m) - 6 > 2\lambda - 4$. If $l', m'$ lie on the same edge of $\T_1$, then $d(l',m') \le \lambda < 2\lambda - 4$, so they have to lie on distinct edges of $\T_1$. If either of $l', m'$ lies on an edge $r_1p_i$, for some $i \le n$, then by construction and the choice of $\rho$, $d(l',m') \le \lambda + 2 < 2\lambda - 4$ (recall that $\lambda > 6$). 
Finally, if either of $l',m'$ lies on $(r_1,r'_1)$ then $d(l',m') \le \rho + \lambda + 2 < 2\lambda - 4$. Thus, both $l'$ and $m'$ lie on distinct bottom edges of $\T_1$. Hence, $\C$ induces a bijection $\chi$ between the leaves of $\T_1$ and $\T_2$, where $\chi(l) = l'$ for $l \in \T_2$ and $l'\in \T_1$ is the leaf whose incident edge contains the image(s) of $l$ under $\C$. Note that if $l_i, l_j \in \T_2$ are incident to $q_i, q_j$ with $q_i \ne q_j$, then $\chi(l_i)$ and $\chi(l_j)$ are incident to $p_{i'}, p_{j'}$ with $p_{i'} \ne p_{j'}$; otherwise $d(l_i,l_j) = 2\lambda + 6$ and $d(\chi(l_i), \chi(l_j)) \le 2 \lambda$, thereby incurring a distortion of at least $6$. Hence, the bijection $\chi$ can be used to partition $X$ into $m$ subsets $X_1,\ldots,X_m$ of equal weight as follows : if $\chi(l) = l'$ for $l$ incident to $q_i$ and $l'$ incident to $p_j$, then $a_j \in X_i$. Thus, $(X,m)$ is a \emph{yes} instance of \textsc{BAL-PART}.\end{proof}

We may also apply the reduction to metric trees with unit edge lengths by subdividing
longer edges with an appropriate number of vertices. We thus have the following theorem.
\begin{theorem}
\label{thm:hardness}
Unless $\mathrm{P=NP}$, there is no polynomial-time algorithm to approximate the Gromov-Hausdorff distance between two metric trees to a factor better than 3, even in the case of metric trees with unit edge lengths.
\end{theorem}

\section{Gromov-Hausdorff and Interleaving Distances}
\label{sec:GH-interleaving}
In this section we show that the GH distance between two tree metric spaces $\T_1$ and  $\T_2$, and the interleaving distance between two appropriately defined trees incuded from $\T_i$s, are within constant factors of each other. 

Given a metric tree $\T = (T,d)$, let $V(T)$ denote the nodes of the tree. 
Given a point $s \in T$ (not necessarily a node), let $f_s : T \rightarrow \mathbb{R}$ be defined as $f_s(x)$~$=$~$-$~$d(s,x)$. Equipped with this function, we obtain a merge tree $T^s$ from $\T$. 
Intuitively, $T^s$ has the structure of rooting $T$ at $s$, and then adding an extra edge incident to $s$ with function value extending from $0$ to $+\infty$. If $s$ is an internal node of $T$ or an interior point of an edge of $T$, $s$ remains the root of $T^s$. But if $s$ is a leaf of $T$, then $s$ gets merged with the infinite edge and the node of $T$ adjacent to $s$ becomes the root of $T^s$.

Let $\T_1 = (T_1, d_1)$ and $\T_2 = (T_2, d_2)$ be two metric trees. Define
\begin{equation}
\label{eq:def_del} \Delta = \min_{u \in V(T_1), v \in V(T_2)} d_I(T_1^u, T_2^v).
\end{equation}
We prove that $\Delta$ is within a constant factor of $d_{GH}(\T_1, \T_2)$. We first prove a lower bound on $\Delta$.
\begin{lemma}
\label{lem:gh_lb} $\frac{1}{2} d_{GH}(\T_1, \T_2) \le \Delta.$
\end{lemma}
\begin{proof}
Suppose $\Delta = d_I(T^s_1, T^t_2)$ for some $s \in V(T_1)$ and $t \in V(T_2)$. Set $f:= f_s$ and $g:= f_t$. Let ${\alpha : T^s_1 \rightarrow T^t_2}, {\beta : T^t_2 \rightarrow T^s_1}$ be $\Delta$-compatible maps. We define the functions $\alpha^* : T_1^s \rightarrow T^t_2$ and $\beta^* : T^t_2 \rightarrow T^s_1$ as follows :
\begin{multicols}{2}
\noindent
\begin{align*}
\alpha^*(x) = 
\begin{cases}
\alpha(x) &\text{if }g(\alpha(x)) \le 0.\\
t &\text{otherwise.}
\end{cases}
\end{align*}
\begin{align*}
\beta^*(y) = 
\begin{cases}
\beta(x) &\text{if }f(\beta(x)) \le 0.\\
s &\text{otherwise.}
\end{cases}
\end{align*}
\end{multicols}
That is, if $\alpha(x)$ is an ancestor of $t$ (resp. $s$) then $x$ (resp. $y$) is mapped to the root $t$ (resp. $s$).
We note that 
\begin{equation}
\label{eq:starDelta}
\begin{aligned}
f(x) \le g(\alpha^*(x)) \le f(x) + \Delta,\\
g(y) \le f(\beta^*(y)) \le g(y) + \Delta.
\end{aligned}
\end{equation}
Indeed, if $\alpha^*(x) = \alpha(x)$ then $g(\alpha^*(x)) = f(x) + \Delta$. Otherwise $g(\alpha(x)) > 0$ and $g(\alpha^*(x)) = 0$. Since $f(x) \le 0$, we obtain $g(\alpha^*(x)) < g(\alpha(x)) = f(x) + \Delta$. The same argument implies the second set of inequalities.

Consider the correspondence $\C \in T_1 \times T_2$ induced by $\alpha^*$ and $\beta^*$ defined as: 
$$\C:= \{(x, \alpha^*(x)) \mid x\in T_1\} \cup \{(\beta^*(y), y) \mid y\in T_2 \}. $$
We prove that $\distortion(\C) \leq 4 \Delta $.
 
Indeed, consider any two pairs $(x_1, y_1), (x_2, y_2) \in \C$. 
Let $u$ be the common ancestor of $x_1$ and $x_2$ in $T_1$, and $w$ the common ancestor of $y_1$ and $y_2$ in $T_2$. 
Note that since $T_1$ and $T_2$ are trees, there is a unique path $x_1 \leadsto u \leadsto x_2$ between $x_1$ and $x_2$, such that $x_1\leadsto u$ and $u \leadsto x_2$ are each monotone in function $f$ values. 
This also implies that $d_1(x_1, u) = d_1 (s, x_1) - d_1(s, u) = f(u) - f(x_1)$; similarly, $d_1(x_2, u) = f(u) - f(x_2)$. 
Symmetric statements hold for $y_1 \leadsto w \leadsto y_2$. 
Hence  
\begin{align*}
d_1 (x_1, x_2) = d_1(x_1, u) + d_1(u, x_2) = 2f(u) - f(x_1) - f(x_2),\\
d_2 (y_1, y_2) = d_2(y_1, w) + d_2(w, y_2) = 2g(w) - g(y_1) - g(y_2).
\end{align*}

We then have,
\begin{align*}
|d_1 (x_1, x_2) - d_2 (y_1, y_2)| &= | 2f(u) - f(x_1) - f(x_2) - 2g(w) + g(y_1) + g(y_2) |\\
&\le 2 |f(u) - g(w)| + |f(x_1) - g(y_1)| + |f(x_2) - g(y_2)|\\
&\le 2|f(u) - g(w))| + 2\Delta \text{    (by \eqref{eq:starDelta})}.
\end{align*}
On the other hand, $\alpha^*(u)$ must be an ancester of $w$, and similarly, $\beta^*(w)$ must be an ancester of $u$.
Thus, $f(u) - \Delta \leq g(w) \leq f(u) + \Delta \Rightarrow |f(u) - g(w)| \leq \Delta$. 
We thus have 
$$ |d_1 (x_1, x_2) - d_2 (y_1, y_2)| \leq 4 \Delta.$$
It then follows that $\distortion(\C) \le 4\Delta$. Since $d_{GH}(\T_1, \T_2) \le \frac{1}{2}\distortion(\C)$, the left inequality then follows.
\end{proof}
Next, we prove an upper bound on $\Delta$.

\begin{lemma}
\label{lem:gh_ub}
$\Delta \le 14 d_{GH} (\T_1, \T_2).$
\end{lemma}
\begin{proof}
Set $\delta = d_{GH}(\T_1, \T_2)$ and let $\C^*: T_1\times T_2$ be an optimal correspondence that achieves  $d_{GH}(\T_1, \T_2)$. Note that in general $d_{GH}(\T_1, \T_2)$ may only be achieved in the limit. In that case, our proof can be modified by considering a sequence of near-optimal correspondences (whose associated metric-distortion converges to $\delta$), and taking a certain limit under it.

Let $s$ be one of the endpoints of a longest simple path in $T_1$ (i.e, the length of this path realizes the diameter of $\T_1$); $s$ is necessarily a leaf of $T_1$. Let $(s,t)$ be a pair in $\C^*$. Consider the merge trees $T_1^s$ and $T_2^t$ defined by the functions $f_s$ and $f_t$, respectively. A result in \cite{dey} implies that $$d_I(T_1^s, T_2^t) \le 6 \delta.$$ We prove below in Claim \ref{claim:c} that there is a vertex (in fact a leaf) $z \in V(T_2)$ such that $d_2(t,z) \le 8 \delta$.

It is easy to verify that $$\|f_t - f_z\|_{\infty} \le d_2(t,z) \le 8 \delta.$$

On the other hand, by the stability theorem of the interleaving distance (Theorem 2 of \cite{morozov}),
$$ d_I(T_2^t, T_2^z) \le \|f_t - f_z\|_{\infty} \le 8\delta.$$

By triangle inequality,
\begin{align*}
d_I(T^s_1, T^z_2) \le& d_I(T^s_1, T^t_2) + d_I(T^t_2, t^z_2)\\
\le& 6 \delta + 8 \delta \\
\le& 14 \delta.
\end{align*}
This completes the proof of the lemma.
\end{proof}

\begin{claim}
\label{claim:c}
Let $s$ be an endpoint of a longest simple path in $T_1$, and let $(s,t)$ be a pair in $\C^*$. Then there is a vertex $z \in V(T_2)$ such that $d_2(t,z) \le 8\delta$.
\end{claim}
\begin{proof}
Assume that there is no tree node within $8\delta$ distance to $t$. In this case, $t$ must be in the interior of an edge $e \in E(T_2)$. 
Let $u_1$ and $u_2$ be the two points in $e$ from opposite sides of $t$ such that $d_2(t, u_1) = d_2(t, u_2) = 8\delta + \nu$, where $\nu > 0$ is an arbitrarily small value. Both $u_1$ and $u_2$ exist, as there is no tree node of $T_2$ within $8\delta$ distance to $t$, and $$d_2(u_1,u_2) = d_2(t, u_1) + d_2(t, u_2) = 16\delta + 2\nu.$$
Let $\tilde{u}_1, \tilde{u}_2 \in T_1$ be any corresponding points for $u_1$ and $u_2$ under $\C^*$, that is, $(\tilde{u}_1, u_1), (\tilde{u}_2, u_2) \in \C^*$. Since $\distortion(\C^*) \le 2\delta$, we have 
\begin{equation}
\label{eq:5}
d_1(\tilde{u}_1, \tilde{u}_2) \ge 14\delta + 2\nu. 
\end{equation}
On the other hand, since $d_2(t, u_1) = d_2(t, u_2) = 8\delta + \nu$, we have that
\begin{equation}
\label{eq:6}
d_1(s,\tilde{u}_1), d_1(s, \tilde{u}_2) \in [6\delta + \nu, 10\delta + \nu].
\end{equation}
We now obtain an upper bound on $d_1(\tilde{u}_1, \tilde{u}_2)$.
\begin{figure}[htb]
\centering
\includegraphics[width=0.15\textwidth]{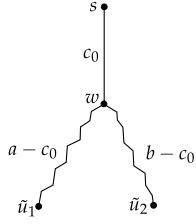}
\caption{\label{fig:ancestor}$w$ is the nearest common ancestor of $\tilde{u}_1$ and $\tilde{u}_2$ in $T^s_1$.}
\end{figure}
If $\tilde{u}_1$ and $\tilde{u}_2$ have ancestor/descendant relation in $T^s_1$, then $d_1(\tilde{u}_1, \tilde{u}_2) = |d_1(s,\tilde{u}_1) - d_2(s,\tilde{u}_2)|$ and by \eqref{eq:6}, we thus have that $d_1(\tilde{u}_1, \tilde{u}_2) \le 4\delta$, which contradicts \eqref{eq:5}.

Now, let $w$ be the nearest common ancestor of $\tilde{u}_1$ and $\tilde{u}_2$ in $T^s_1$ (see Figure~\ref{fig:ancestor}). Let $c_0 = d_1(s,w)$. For simplicity, set $a = d_1(s,\tilde{u}_1)$ and $b = d_1(s, \tilde{u}_2)$. It then follows that
\begin{equation}\label{eq:7}
d_1(\tilde u_1, \tilde u_2) = a + b - 2c_0\enspace. ~~\text{Note,}~~a \ge c_0, b\ge c_0\enspace. 
\end{equation}
Since $s$ is an endpoint of the longest path in $T_1$, it follows that $c_0 \ge \min\{a-c_0, b-c_0\}$ (if not, then without loss of generality, suppose the other point $s'$ of the diameter pair is not in the subtree of $T_1^s$ rooted at $\tilde{u_1}$; then $d_1(\tilde{u_1}, s') > d_1(s,s')$, a contradiction). By \eqref{eq:6}, $a,b \ge 6\delta +\nu$. Thus
\begin{align}\label{eq:8}
c_0 &\ge 6\delta + \nu - c_0 ~~\Rightarrow c_0 \ge \frac{1}{2}[6\delta + \nu]\enspace.
\end{align}
Combining \eqref{eq:7} and \eqref{eq:8}, we have
\begin{equation}
\label{eq:9}
d_1(\tilde{u}_1, \tilde{u}_2) \le a+b - 6\delta - \nu \le 20\delta + 2\nu -6\delta -\nu = 14 \delta + \nu,
\end{equation}
contradicting \eqref{eq:5}. Thus, there exists $z \in V(T_2)$ such that $d_2(t,z) \le 8\delta$.
\end{proof}

\fakeparagraph{Remark.} The proof of Claim~\ref{claim:c} actually shows that $t$ lies in the neighborhood of a leaf, as we never use the fact that $u_1$ and $u_2$ lie on the same edge of $T_2$. The only fact we use is that $u_1$ and $u_2$ lie on opposite sides of $t$ at distance $8\delta + \nu$ each.

From Lemmas~\ref{lem:gh_lb} and \ref{lem:gh_ub}, we get the following.

\begin{theorem}
\label{thm:gh_i}
Let $\Delta = \min_{u \in V(T_1), v \in V(T_2)} d_I(T^u_{1}, T^v_{2})$. Then 
\begin{align*}
\tfrac{1}{2} d_{GH}(\T_1, \T_2) \leq \Delta \leq 14 d_{GH}(\T_1, \T_2).
\end{align*}
\end{theorem}

\begin{corollary}
\label{cor:gh_i}
If there is a polynomial time, $c$-approximation algorithm for the interleaving distance between two merge trees, then there is a polynomial time, $28c$-approximation algorithm for the Gromov-Hausdorff distance between two metric trees.
\end{corollary}


\section{Computing the Interleaving Distance}
\label{sec:compute_id}
Let $\mt_f$ and $\mt_g$ be merge trees of two functions $f$ and $g$, respectively. For simplicity, we use $f$ and $g$ to denote the height functions  on $\mt_f$ and $\mt_g$ as well. Let $n$ be the total number of nodes in $\mt_f$ and $\mt_g$, and let $r \ge 1$ be the ratio between the lengths of the longest and the shortest edges in $\mt_f$ and $\mt_g$. We describe a $O(\min\{n, \sqrt{rn}\})$-approximation algorithm for computing $d_I(\mt_f, \mt_g)$.

\fakeparagraph{Candidate values and binary search.} We first show that a candidate set $\Lambda$ of $O(n^2)$ values can be computed in $O(n^2)$ time such that $d_I(\mt_f,\mt_g) \in \Lambda$. Given $\Lambda$, we perform a binary search on $\Lambda$. At each step, we use a $c$-approximate decision procedure, for $c = c_1 \min\{n, \sqrt{rn}\}$ for some constant $c_1$, that given a value $\eps > 0$ does the following~: if $d_I(\mt_f, \mt_g) \leq \eps$, it returns a pair of $c\eps$-compatible maps between $\mt_f$ and $\mt_g$; if $d_I(\mt_f, \mt_g) > \eps$, it will either return a pair of $c\eps$-compatible maps between $\mt_f$ and $\mt_g$ or report that no such maps exist. The binary search terminates when one of the following two conditions meet :
\begin{itemize}
\item[(i)] We have two consecutive values $\eps^-, \eps^+ \in \Delta$ with $\eps^- < \eps^+$ such that the decision procedure returned \textsc{Yes} for $\eps^+$ and \textsc{No} for $\eps^-$; in this case we return $\eps^+$.
\item[(ii)] We have two (not necessarily consecutive) values $\eps^-, \eps^+ \in \Delta$ with $\eps^- < \eps^+$ such that the decision procedure returned \textsc{No} for $\eps^+$ and $\textsc{Yes}$ for $\eps^-$ (but with $\eps'$-compatible maps for some $\eps' > \eps^+$), in which case we return $\eps^-$.
\end{itemize}

It is clear that the procedure returns a value $\eps$ such that $d_I(\mt_f, \mt_g) \le \eps \le c d_I(\mt_f, \mt_g)$. We now describe the candidate set $\Lambda$. 

Let $V_f$ (resp. $V_g$) be the set of nodes in $\mt_f$ (resp. $\mt_g$). We define $\Lambda = \Lambda_{11} \cup \Lambda_{22} \cup \Lambda_{12}$, where
\begin{align*}
\Lambda_{11} &= \{\tfrac{1}{2} |f(u) - f(v)| \mid u,v \in V_f \},\\
\Lambda_{22} &= \{\tfrac{1}{2} |g(u) - g(v)| \mid u,v \in V_g\},\\
\Lambda_{12} &= \{|f(u) - g(v)| \mid u \in V_f, v \in V_g\}.
\end{align*}

\begin{lemma}
\label{lem:delta}
$d_I(\mt_f, \mt_g) \in \Lambda.$
\end{lemma}
\begin{proof}
Suppose to the contrary that $d_I(\mt_f, \mt_g) = \eps \notin \Lambda$. Let $\alpha : \mt_f \rightarrow \mt_g$ and $\beta : \mt_g \rightarrow \mt_f$ be $\eps$-compatible maps that realize $d_I(\mt_f, \mt_g) = \eps$. We will obtain a contradiction by choosing $\eps_0 > 0$ and constructing $(\eps - \eps_0)$-compatible maps $\hat{\alpha}, \hat{\beta}$.

For any point $x \in \mt_f$, we define $\alpha_{\downarrow}(x) = \alpha(x)$ if $\alpha(x)$ is a node of $\mt_g$, otherwise $\alpha_{\downarrow}(x)$ is the lower endpoint of the edge of $\mt_g$ containing $\alpha(x)$. Similarly we define the  function $\beta_{\downarrow} : \mt_g \rightarrow \mt_f$.

\begin{figure}[htb]
\centering
\includegraphics[width=0.45\textwidth]{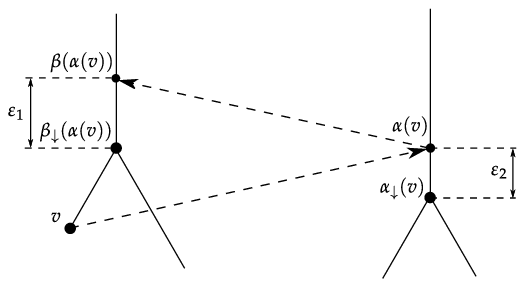}
\caption{\label{fig:delta_v}Trees $\mt_f$ and $\mt_g$. Here $\delta_v = \min\{\tfrac{1}{2}\eps_1, \eps_2\}$.}
\end{figure}

For every node $v \in V_f$, $\alpha(v)$ (resp. $\beta(\alpha(v))$) lies in the interior of an edge of $\mt_g$ (resp. $\mt_f$), because $\eps \notin \Lambda \supseteq \Lambda_{12}$ (resp. $\Lambda_{11}$). We define
$$ \delta_v = \min\{\tfrac{1}{2}\left(f(\beta(\alpha(v))) - f(\beta_{\downarrow}(\alpha(v))), g(\alpha(v)) - g(\alpha_{\downarrow}(v))\right)\}.$$
See Figure~\ref{fig:delta_v}. Similarly we define $\delta_w$ for all $w \in V_g$. We set
$$\eps_0 = \min\{\eps, \min_{v \in V_f \cup V_g}\delta_v\}.$$
Since $\eps \notin \Lambda$, we have $\eps_0 > 0$. We now construct $(\eps - \eps_0)$-compatible maps $\hat{\alpha} : \mt_f \rightarrow \mt_g$ and $\hat{\beta} : \mt_g \rightarrow \mt_f$. We describe the construction of $\hat{\alpha}$; $\hat{\beta}$ is constructed similarly. By construction, for any node $u \in V_f$, $g(\alpha(u)) - g(\alpha_{\downarrow}(u)) \ge \eps_0$, so we set $\hat{\alpha}(v)$ to be the point $w$ on the edge of $\mt_g$ containing $\alpha(u)$ such that $g(w) = f(u) + \eps - \eps_0$. Once we have defined $\hat{\alpha}(u)$ and $\hat{\alpha}(v)$ for an edge $uv \in \mt_f$, with $f(u) < f(v)$, we set $\hat{\alpha}(x)$, for a point $x \in uv$ with $f(x) = f(u) + \gamma$, to be
$$\hat{\alpha}(x) = \shift_g^{\gamma}(\hat{\alpha}(u)).$$
That is, we set $\hat{\alpha}(x)$ to be the ancestor of $\hat{\alpha}(u)$ at height $f(u) + \eps - \eps_0 + \gamma = f(x) + \eps - \eps_0$.

We claim that $\hat{\alpha}, \hat{\beta}$ are $(\eps - \eps_0)$-compatible. Indeed, by construction, $g(\hat{\alpha}(x)) = f(x) + \eps - \eps_0$ for all $x \in \mt_f$, and $f(\hat{\beta}(y)) = g(y) + \eps -\eps_0$ for all $y \in \mt_g$. We now prove that $$\hat{\beta} \circ \hat{\alpha} = \shift_f^{2(\eps-\eps_0)}.$$
\begin{figure}[htb]
\centering
\includegraphics[width=0.22\textwidth]{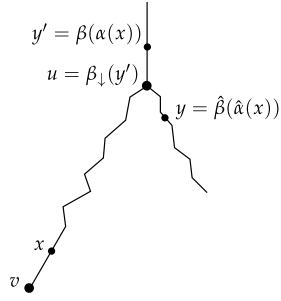}
\caption{\label{fig:candidates}Figure showing $u,v,x,y$ and $y'$.}
\end{figure}
Suppose to the contrary there is a point $x \in \mt_f$ such that $y = \hat{\beta}(\hat{\alpha}(x)) \ne \shift_f^{2(\eps-\eps_0)}(x)$. Since $f(y) = f(x) + 2(\eps-\eps_0)$, $y$ must not be an ancestor of $x$. On the other hand, $\alpha, \beta$ are $\eps$-compatible, so $y' = \beta(\alpha(x))$ is the ancestor of $x$ at height $f(x) + 2 \eps$. By construction of $\hat{\alpha}$ and $\hat{\beta}$, $y$ is a descendant of $y'$, in which case there is a node $u \in V_f$ that lies between $y$ and $y'$. (If $y$ and $y'$ lie on the same edge of $\mt_f$, then $y$ is also an ancestor of $x$.) Let $u = \beta_{\downarrow}(\alpha(x))$. Let $v$ be the lower endpoint of the edge $e$ containing $x$. See Figure~\ref{fig:candidates}. Since $\eps \notin \Delta$, $f(u) \ne f(v) + 2 \eps$ (i.e., $u \ne \beta(\alpha(v))$). There are two cases to consider :
\begin{itemize}
\item[(i)] $f(u) > f(v) + 2 \eps$. Then let $u = \beta(\alpha(z))$ for the point $z$ lying between $x$ and $v$ at height $f(z) = f(u) - 2\eps$. Furthermore $f(x) \ge f(z) > f(x) - 2\eps_0$ (if $f(x) - f(z) \ge 2\eps_0$, then $f(y') - f(u) = f(x) - f(z) \ge 2\eps_0$, contradicting the fact that $f(y') - f(y) = 2\eps_0$). Therefore we can choose a point $w \ne v$ on $e$ such that $f(z) > f(w) > f(x) - 2\eps_0$. Now, it's not too hard to see that if $x_1$ is an ancestor of $x_2$ in $\mt_f$, then $\hat{\alpha}(x_1)$ is an ancestor of $\hat{\alpha}(x_2)$ (similarly for $\hat{\beta}$). Further, $\hat{\beta}(\hat{\alpha}(x_1))$ is a descendant of $\beta(\alpha(x_1))$ for all $x_1 \in \mt_f$ (a similar result holds for $\hat{\alpha} \circ \hat{\beta}$ and $\alpha \circ \beta$). Thus, $\hat{\beta}(\hat{\alpha}(w))$ is a descendant of $y = \hat{\beta}(\hat{\alpha}(x))$ (since $w$ is a descendant of $x$). Moreover, $\beta(\alpha(w))$ is an ancestor of $\hat{\beta}(\hat{\alpha}(w))$. However, since $f(\beta(\alpha(w))) < f(u)$,  $\beta(\alpha(w))$ lies between $y$ and $u$. Thus, $\beta(\alpha(w))$ is not an ancestor of $x$ (hence $w$), i.e., $\beta(\alpha(w)) \ne \shift_f^{2\eps}(w)$, contradicting the fact that $\alpha, \beta$ are $\eps$-compatible.
\item[(ii)] $f(u) < f(v) + 2 \eps$. In this case
$$f(\beta(\alpha(v))) - f(\beta_{\downarrow}(\alpha(v))) \le f(\beta(\alpha(v))) - f(u) < f(y') - f(y) = 2\eps_0 \le 2\delta_v,$$
which contradicts the definition of $\delta_v$.
\end{itemize}
Hence we conclude that $y$ is an ancestor of $x$, i.e., $\hat{\beta} \circ \hat{\alpha} = \shift_f^{2(\eps-\eps_0)}$. Similarly, we argue that $\hat{\alpha} \circ \hat{\beta} = \shift_g^{2(\eps-\eps_0)}$, implying that $\hat{\alpha}, \hat{\beta}$ are $(\eps - \eps_0)$-compatible maps as claimed.

Putting everything together, we conclude that $\eps \in \Delta$.
\end{proof}

We now describe the decision procedure to answer the question ``is $d_I(\mt_f, \mt_g) \le \eps$?" approximately. 
We define the \emph{length} of any edge in a merge tree (other than the edge to infinity) to be the
height difference between its two endpoints. Given a parameter $\eps > 0$, an edge is called
$\eps$-\emph{long}, or \emph{long} for brevity, if its length is strictly greater than $2\eps$. We first describe an exact decision procedure for the case when all edges in both trees are long, and then describe an approximate decision procedure for the case when the two trees have short edges.

\fakeparagraph{Trees with long edges.}
\label{exact}
We remove all degree-two nodes in the beginning. A \emph{subtree} rooted at a point $x$ in a merge tree $\mt$, denoted $\mt^x$, includes all the points in the merge tree that are descendants of $x$ and an edge from $x$ that extends upwards to  height $\infty$. For a node $u \in V$, let $C(u)$ denote the children of $u$ and let $p(u)$ denote its parent.
Assume $d_I(\mt_f,\mt_g) \le \eps$, and let $\alpha : \mt_f \rightarrow \mt_g$ and $\beta : \mt_g \rightarrow \mt_f$ be a pair of $\eps$-compatible maps. As in the proof of Lemma~\ref{lem:delta}, we define the functions $\alpha_{\downarrow}$ and $\beta_{\downarrow}$ but restricted only to the vertices of $\mt_f$ and $\mt_g$. That is, for a node $v \in V_f$, we define $\alpha_{\downarrow}(v)$ to be the lower endpoint of the edge containing $\alpha(v)$ -- if $\alpha(v)$ is a node, then $\alpha_{\downarrow}(v)$ is $\alpha(v)$ itself. Similarly we define $\beta_{\downarrow}(w)$, for a node $w \in V_g$.

The following two properties of $\alpha_{\downarrow}$ and $\beta_{\downarrow}$ will be crucial for the decision procedure.

\begin{lemma}
\label{lem:A}
(i) For a node $v \in V_f$, $|f(v) - g(\alpha_{\downarrow}(v))| \le \eps$, and 
(ii) for a node $w \in V_g$, $|g(w) - f(\beta_{\downarrow}(w))| \le \eps$.
\end{lemma}
\begin{proof}
We will prove part (i); part (ii) is similar. By definition, $g(\alpha_{\downarrow}(v)) \le f(v) + \eps$. Suppose $g(\alpha_{\downarrow}(v)) < f(v) - \eps$. Let $v'$ be a point in $\mt_g$ lying on the edge containing $\alpha(v)$ and $\alpha_{\downarrow}(v)$ with height $f(v) - \eps - \eps_0$, for some sufficiently small $\eps_0$. Then $\beta(v')$ lies in one of the subtrees  rooted at the children of $v$, say $\mt_1$. Consider a descendant $u$ of $v$ at height $g(v') - \eps$ lying in a different subtree $\mt_2$ rooted at $v$'s child. Since by definition and our choice of $v'$ there does not exist any node in $\mt_g$ between $\alpha(v)$ and $v'$, we have $\alpha(u) = v'$. But then $\beta(\alpha(u)) = \beta(v')$ lies in $\mt_1$, and hence is not an ancestor of $u \in \mt_2$; in other words $\beta(\alpha(u)) \ne \shift_f^{2\eps}(u)$. This contradicts the fact that $\alpha$ and $\beta$ are $\eps$-compatible. Thus, $g(\alpha_{\downarrow}(v)) \ge f(v) - \eps$, and the claim follows.
\end{proof}

\begin{lemma}
\label{lem:B}
If all edges in $\mt_f$ and $\mt_g$ are $\eps$-long, then $\alpha_{\downarrow}$ and $\beta_{\downarrow}$ are bijections with $\beta_{\downarrow} = \alpha_{\downarrow}^{-1}$ (and $\alpha_{\downarrow} = \beta_{\downarrow}^{-1}$).
\end{lemma}
\begin{proof}
We will first show that  $\beta_{\downarrow} = \alpha_{\downarrow}^{-1}$. Suppose to the contrary there exists a vertex $v \in V_f$ such that $\beta_{\downarrow}(\alpha_{\downarrow}(v)) = w \ne v$. Let $\alpha_{\downarrow}(v) = u$, for $u \in V_g$. From Lemma~\ref{lem:A} we have $|f(v) - f(w)| \le 2 \eps$. Since all edges are longer than $2\eps$ and $v \ne w$, $v$ cannot be an ancestor/descendant of $w$ in $\mt_f$. By definition of $\alpha_{\downarrow}$, $\alpha(v)$ is an ancestor of $\alpha_{\downarrow}(v) = u$. Thus $\beta(\alpha(v))$ is an ancestor of $\beta(u)$. Further, $|f(v) - g(u)| \le \eps$ (Lemma~\ref{lem:A}) and $\beta(\alpha(v)) = \shift_f^{2\eps}(v)$ (since $\alpha,\beta$ are $\eps$-compatible). 
Hence, $\beta(u)$ lies between $v$ and $\beta(\alpha(v)$ on the edge $e$ whose lower endpoint is $v$ as $e$ is $\eps$-long. 
Thus, $\beta(u)$ is an ancestor of $v$.
Also by definition of $\beta_{\downarrow}$, $\beta(u)$ is an ancestor of $\beta_{\downarrow}(u) = w$. Thus, $w$ is a descendant of $v$, a contradiction since $v$ cannot be an ancestor of $w$.

We thus have  $\beta_{\downarrow} = \alpha_{\downarrow}^{-1}$. Similarly, we can show that  $\alpha_{\downarrow} = \beta_{\downarrow}^{-1}$. This also implies that $\alpha_{\downarrow}$ and $\beta_{\downarrow}$ are bijections.
\end{proof}

We define an indicator function $\Phi : V_f \times V_g \rightarrow \{0,1\}$ such that 
\begin{align*}
\Phi(u,v) =
\begin{cases}
1, &\text{if } d_I(\mt_f^u, \mt_g^v) \le \eps,\\
0, &\text{otherwise.}
\end{cases}
\end{align*}
The following lemma gives a recursive definition of $\Phi(u,v)$.

\begin{lemma}
\label{lem:subproblem}
Suppose all the edges in $\mt_f$ and $\mt_g$ are $\eps$-long.
For a pair $(u,v) \in V_f \times V_g$, $\Phi(u,v) = 1$ if and only if the following conditions hold : (i) $|f(u) - g(v)| \leq \eps$, (ii)
$|C(u)| = |C(v)|$, and (iii) there exists a permutation~$\pi$ of $[1:|C(u)|]$ such that $\Phi(u_i,
v_{\pi(i)}) = 1$ for all $i \in [1:|C(u)|]$.
\end{lemma}
\begin{proof}
Suppose $\Phi(u,v) = 1$, and let $\alpha,\beta$ be the corresponding $\eps$-compatible maps.
To see why (i) holds, for contradiction, suppose property (i) does not hold, and let~$f(u) > g(v)$ without loss of
generality.
Thus, $\beta(v)$ maps to one of the multiple edges incident to $u$, and there exists at least one edge $e = (u,w)$ with $w \in C(u)$ such that none of $e$'s points (other than $u$) is in the image of $\beta$. However, $\beta(\alpha(u)) = \shift_f^{2\eps}(u)$ must lie in the interior of $e$ (since $e$ is $\eps$-long), a contradiction. To prove that (ii) holds, note that by Lemma~\ref{lem:B}, there exist bijections $\alpha_{\downarrow}, \beta_{\downarrow}$ between $V_f, V_g$ such that if $u_1 \in C(u_2)$ in $V_f$, then $\alpha_{\downarrow}(u_1) \in C(\alpha_{\downarrow}(u_2))$ (a symmetric statement holds for $\beta_{\downarrow}$ and $V_g$). Thus, $\alpha_{\downarrow}, \beta_{\downarrow}$ induce bijections between $C(u)$ and $C(v)$, and hence $|C(u)| = |C(v)|$. 
Finally, for (iii), Let $\alpha_{\downarrow}(u') = v'$ for some $u' \in C(u), v' \in C(v)$. Then by definition of $\alpha_{\downarrow}$ and $\beta_{\downarrow}$, $\alpha(\mt_f^{u'}) \subseteq \mt_g^{v'}$ and $\beta(\mt_g^{v'}) \subseteq \mt_f^{u'}$. 
This means that the restriction of the pair of $\eps$-compatible maps $\alpha$ and $\beta$ to $\mt_f^{u'}$ and $\mt_g^{v'}$ respectively remain $\eps$-compatible for $\mt_f^{u'}$ and $\mt_g^{v'}$. 
Thus, $\Phi(u',v') = 1$, and the permutation $\pi$ is defined by $\alpha_{\downarrow}, \beta_{\downarrow}$.

We now prove the opposite direction. Suppose properties (i),(ii) and (iii) hold. Let $(\alpha_i, \beta_i)$ be the pair of
$\eps$-compatible maps between $\mt_f^{u_i}$ and $\mt_g^{v_{\pi(i)}}$. Then, a pair of
$\eps$-compatible maps $(\alpha, \beta)$ between $\mt_f^u$ and $\mt_g^v$ is obtained as follows :
$\alpha(x) = \{\alpha^i(x) \mid x \in \mt_f^{u_i}\}$ ($\beta$ is defined similarly). Note that
points on the infinite edge from $u$ (resp. $v$) upwards are \emph{shared} among all $\mt_f^{u_i}$ (resp.
$\mt_g^{v_j}$), whereas all other points in $\mt_f^u$ (resp. $\mt_g^v$) are present in only one
$\mt_f^{u_i}$ (resp. $\mt_g^{v_j}$). However, since $|f(u)-g(v)| \leq \eps$, \emph{shared} points
are mapped to \emph{shared} points and we have $|\alpha(x)| = 1$ (resp. $|\beta(y)| = 1|$) for all
$x \in \mt_f^u$ (resp. $y \in \mt_g^v$). Thus, $\alpha$ and $\beta$ are functions and satisfy all the
required properties. Hence, $\Phi(u,v) = 1$.
\end{proof}

$\mathbf{\mathit{Decision}}$ $\mathbf{\mathit{procedure.}}$ We compute $\Phi$ for all pairs of nodes in $V_f \times V_g$ in a bottom-up manner and return $\Phi(r_f, r_g)$ where $r_f$ (resp. $r_g$) is the root of $\mt_f$ (resp. $\mt_g$). Let $(u,v) \in V_f \times V_g$.

Suppose we have computed $\Phi(u_i,v_j)$ for all $u_i \in C(u)$ and $v_j \in C(v)$. We compute $\Phi(u,v)$ as follows. If (i) or (ii) of Lemma~\ref{lem:subproblem} does not hold for $u$ and $v$, then we return $\Phi(u,v) = 0$. Otherwise we construct the bipartite graph $G_{uv} = \{C(u) \cup C(v), E = \{(u_i,v_j) \mid \Phi(u_i,v_j) = 1\}\}$ and determine in $O(k^{5/2})$ time whether $G_{uv}$ has a perfect matching, using the algorithm by Hopcroft and Karp \cite{bipartite}. Here, $k = |C(u)| = |C(v)|$. If $G_{uv}$ has a perfect matching $M = \{(u_1, v_{\pi(1)}), \ldots, (u_k, v_{\pi(k)})\}$, we set $\Phi(u,v) = 1$, else we set $\Phi(u,v) = 0$. If $\Phi(u,v) = 1$, we use the $\eps$-compatible maps for $\mt_f^{u_i}, \mt_g^{v_{\pi(i)}}$, for $1 \leq i \leq k$, to compute a pair of $\eps$-compatible maps between $\mt_f^u$ and $\mt_g^v$, as discussed in the proof of Lemma~\ref{lem:subproblem}.

For a node $u \in V_f \cup V_g$, let $k_u$ be the number of its children.
The total time taken for running Hopcroft and Karp~\cite{bipartite} is :
\begin{align*}
  \sum_{u \in V(T_1)} \sum_{v \in V(T_2)} O\left(k_u k_v \sqrt{k_v}\right) = \sum_{u \in V(T_1)} k_u \sum_{v \in
  V(T_2)} O\left(k_v \sqrt{k_v}\right)
\leq O\left(n^{3/2}\right) \sum_{u \in V(T_1)} k_u
\leq O\left(n^{5/2}\right).
\end{align*}
Hence we obtain the following.
\begin{lemma}
\label{lem:easy_case}
Given two merge trees $\mt_f$ and $\mt_g$ and a parameter $\eps > 0$ such that all edges of $\mt_f$ and $\mt_g$ are $\eps$-long, then whether $d_I(\mt_f,\mt_g) \leq \eps$ can be determined in $O(n^{5/2})$ time. If the answer is yes, a pair of $\eps$-compatible maps between $\mt_f$ and $\mt_g$ can be computed within the same time.
\end{lemma}
\begin{figure}[htb]
 \centering
  \includegraphics[height=0.2\textheight]{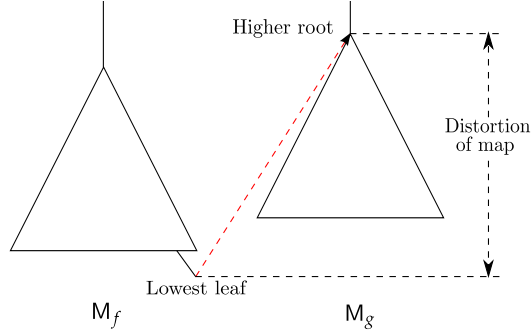}
\caption{\label{naive_map}A naive map.}
\end{figure}

\fakeparagraph{Trees with short edges.}
Given two merge trees, a naive map is to map the lowest among all the leaves in both the trees to
a point at height equal to the height of the higher root (see Figure~\ref{naive_map}). Thus, all the points in one tree will be mapped to the infinitely long edge on the other tree.
This map produces a distortion equal to the height of the trees, which can be arbitrarily larger than the optimum. Nevertheless, this simple idea leads to an approximation algorithm.

Here is an outline of the algorithm. After carefully \emph{trimming} off short subtrees from the input
trees, the algorithm decomposes the resulting trimmed trees into two kinds of regions -- those with nodes and those without nodes. If the interleaving distance between the input trees is small, then there exists an isomorphism between trees induced by the regions without nodes.
Using this isomorphism, the points in the nodeless regions are mapped without incurring additional
distortion. 
Using a counting argument and the naive map described above, it is shown that the distortion
incurred while mapping the regions with nodes and the trimmed regions is bounded.

\begin{figure}[htb]
\centering
\includegraphics[width=0.4\textwidth]{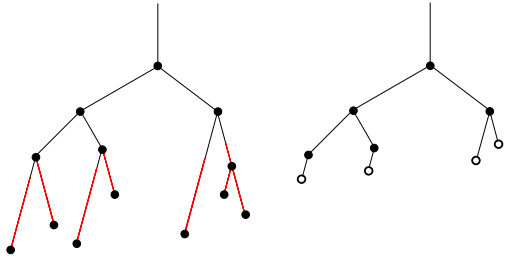}
\caption{\label{fig:extent}Trimming a tree : (left) original tree, red points have extent $< 2(\sqrt{2ns} + 1)\eps$; (right) trimmed tree, nodes added at the bottom (hollowed nodes). }
\end{figure}

More precisely, given $\mt_f$, $\mt_g$ and $\eps > 0$, define the \emph{extent $e(x)$} of a \emph{point}~$x$ (which is not necessarily a tree node) in $\mt_f$ or $\mt_g$ as the maximum height difference between $x$ and any of its descendants. Suppose each edge is at most $s\eps$ long. Let $\mt'_f$ and $\mt_g'$ be subsets of $\mt_f$ and $\mt_g$ consisting only of points with extent at least $2(\sqrt{2ns}+1)\eps$, adding nodes to the new leaves of $\mt'_f$ and $\mt'_g$ as necessary. 
Note that $\mt_f'$ and $\mt_g'$ themselves are trees, however they might contain nodes of degree 2. See Figure~\ref{fig:extent} for an example.
\begin{lemma}
\label{lem_extent}
If $d_I(\mt_f,\mt_g) \leq \eps$, then $d_I(\mt_f', \mt_g') \leq \eps$.
\end{lemma}
\begin{proof}
Let~$\alpha : \mt_f \to \mt_g$ and $\beta: \mt_g \to \mt_f$ be $\eps$-compatible maps. Let~$\alpha'$ and~$\beta'$ be restrictions of the functions' domains to $\mt_f'$ and $\mt_g'$ respectively. We argue that the ranges of~$\alpha'$ and~$\beta'$ lie in $\mt_g'$ and $\mt_f'$ respectively. Suppose otherwise. Then without loss of generality, there is a point~$x \in \mt_f'$ with~$y = \alpha(x)$ not in $\mt_g'$. Because $x \in \mt_f'$, its extent in~$\mt_f$ is at least $2(\sqrt{2ns}+1)\eps$. Therefore, there exists a descendant~$x'$ of~$x$ in~$\mt_f$ with $f(x') = f(x) - 2(\sqrt{2ns}+1)\eps$. Because~$y$ is not in~$\mt_g'$, the extent of $y$ must be less than $2(\sqrt{2ns}+1)\eps$ and there exists no descendant $y'$ of $y$ with $g(y') = g(y) - 2(\sqrt{2ns}+1)\eps = f(x) - 2(\sqrt{2ns}+1)\eps + \eps = f(x') + \eps$. Since $g(\alpha(x')) = f(x') + \eps$, $\alpha(x')$ is not a descendant of $\alpha(x)$, which contradicts the assumption that $\alpha, \beta$ are $\eps$-compatible maps.
\end{proof}
The above lemma can be easily generalized to say that removing points in both trees with extent less than or equal to any fixed value does not change the distance between them.

We now define \emph{matching points} in $\mt_f'$ and $\mt_g'$. A \emph{branching node} is a node of degree greater than 2. A point $x$ in $\mt_f'$ is a matching point if there exists a branching node or a leaf $x'$ in $\mt_f'$ or $y'$ in $\mt_g'$ with function value $f(x)$ and there exist no branching nodes nor leaves in $\mt_f'$ or $\mt_g'$ with function value in the range $(f(x), f(x) + 2\eps]$. Matching points on $\mt_g'$ are defined similarly. By this definition, no two matching points share a function value within $2\eps$ of each other unless they share the exact same function value. Furthermore, if $x$ is a matching point, then all points with the same function value as $x$ on both $\mt_f'$ and $\mt_g'$ are matching points.
There are~$O(n^2)$ matching points. 

Suppose $d_I(\mt_f', \mt_g') \leq \eps$, and let $\alpha' : \mt_f' \to \mt_g'$ and $\beta': \mt_g' \to \mt_f'$ be a pair of $\eps$-compatible maps. 
Call a matching point $x$ in $\mt_f'$ and a matching point $y$ in $\mt_g'$ with $f(x) = g(y)$ \emph{matched} if $\alpha'(x)$ is an ancestor of $y$.
\begin{lemma}
\label{lem:matching}
Let $x$ be any matching point in $\mt_f'$. The matched relation between matching points in $\mt_f'$ at height $f(x)$ and matching points in $\mt_g'$ at height $f(x)$ is a bijective function.
\end{lemma}
\begin{proof}
No two distinct matching points $y_1$ and $y_2$ on $\mt_g'$ with $f(x) = g(y_1) = g(y_2)$ share the same ancestor with function value $f(x) + \eps$, because they have no branching node ancestors with low enough function value. Therefore, a matching point in $\mt_f'$ can be matched to only one one matching point in $\mt_g'$. 

Let $x_1$ and $x_2$ be two distinct matching points from $\mt_f'$ with $f(x) = f(x_1) = f(x_2)$. If $\alpha'(x_1)$ and $\alpha'(x_2)$ are ancestors of a common matching point $y$, then $\alpha'(x_1) = \alpha'(x_2)$ and thus $x_1$ and $x_2$ must have a common ancestor $x'$ at height $f(x) + 2\eps$. However, $x_1$ and $x_2$ have no branching node ancestor with low enough function value for $x'$ to exist. Hence, the matching relation must be injective from matching points in $\mt'_f$ to $\mt'_g$.

Finally, consider any matching point $y$ on $\mt_g'$ with $g(y) = f(x)$. Point $x_1' = \beta'(y)$ is the ancestor of a matching point $x_1$ on $\mt_f'$. (Note that by the same argument as the beginning of this proof, only one such matching point $x_1$ can exist.) Point $y' = \alpha'(x_1')$ is an ancestor of $y$ with $g(y') \leq g(y) + 2\eps$. Point $y$ is the only descendant of $y'$ with function value $f(x)$. Point $\alpha'(x_1)$ must be an ancestor of $y$, meaning $x_1$ and $y$ are matched. Thus, the matching relation is surjective.
\end{proof}

\begin{figure}[htb]
\centering
\includegraphics[width=0.5\textwidth]{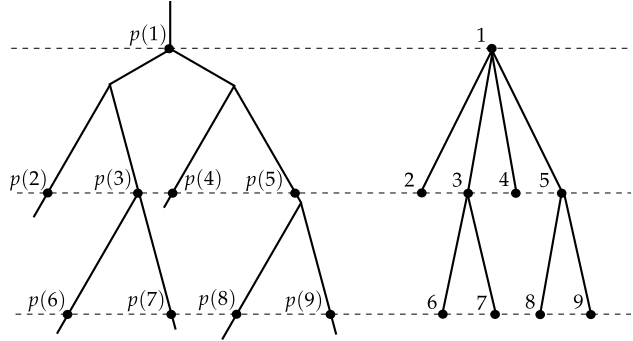}
\caption{\label{fig:matching}The left tree shows matching points on tree $\mt_f'$ and the right tree shows $\tilde{\mt}_f$. }
\end{figure}

We now define a rooted tree $\tilde{\mt}_f$ to be a rooted tree consisting of one node per matching point on $\mt_f'$. Let $p(v)$ be the matching point for node $v$. $\tilde{\mt}_f$ has node $v$ as an ancestor of node $u$ if $p(v)$ is an ancestor of $p(u)$ (see Figure~\ref{fig:matching}). Define $\tilde{mt}_g$ similarly. The size of $\tilde{\mt}_f$ and $\tilde{\mt}_g$ is $O(n^2)$. Intuitively, $\tilde{\mt}_f$ and $\tilde{\mt}_g$ represent the trees induced by matching points. By the definition of interleaving distance and Lemma~\ref{lem:matching}, $\tilde{\mt}_f$ and $\tilde{\mt}_g$ are isomorphic if $\mt'_f$ and $\mt'_g$ satisfy that $d_I (\mt'_f, \mt'_g) \le \eps$. 

$\mathit{Decision}$ $\mathit{procedure}$. We are now ready to describe the decision procedure. We first construct the subtrees $\mt'_f$ and $\mt'_g$ of $\mt_f$ and $\mt_g$, respectively, consisting of points with extent at least $2(\sqrt{2ns} + 1)\eps$. Next, we compute matching points on $\mt'_f$ and $\mt'_g$ and construct the trees $\tilde{\mt}_f$ and $\tilde{\mt}_g$ on these matching points, as defined above.

Using the algorithm of~\cite[chap. 3, p. 85]{algo}, we determine in time linear in the size of the trees whether $\tilde{\mt}_f$ and $\tilde{\mt}_g$ are isomorphic. If the answer is no, we return no. By Lemma~\ref{lem:matching}, $d_I(\mt_f, \mt_g) > \eps$ in this case. Otherwise we construct the following functions $\alpha: \mt_f \to \mt_g$ and $\beta: \mt_g \to \mt_f$ and return them.
For each pair of matching points $x$ and $y$ matched by the isomorphism, the algorithm sets $\alpha(x) = y$ and $\beta(y) = x$.
Now, let $(\xi_1, \xi_2)$ be any maximal range of function values without any branching nodes or leaves in $\mt_f'$ or $\mt_g'$ with $\xi_2 - \xi_1 > 2\eps$. Let $x'$ be any point in $\mt_f'$ with $f(x') \in (\xi_1, \xi_2)$. Point $x'$ has a unique matching point descendant $x$ at height $\xi_1$, by the definition of matching points. The algorithm sets $\alpha(x')$ to the point $y'$ in $\mt_g'$ where $y'$ is the ancestor of $\alpha(x)$ with $g(y') = f(x')$, and it sets $\beta(y') = x'$. For every remaining point $x''$ in $\mt_f'$, the algorithm sets $\alpha(x'')$ to $\alpha(x)$ where $x$ is the lowest matching point \emph{ancestor} of $x''$. $\beta(y'')$ is defined similarly for remaining points $y''$ in $\mt_g'$. We call such points $x''$ and $y''$ \emph{lazily assigned}. Finally, each point $z$ in $\mt_f - \mt_f'$ has $\alpha(z)$ set to $\alpha(x)$ where $x$ is the lowest ancestor of $z$ in $\mt_f'$. Similar assignments are done for points in $\mt_g - \mt_g'$.


\begin{lemma}
\label{lem:counting}
(i) For each lazily assigned point $x'' \in \mt_f'$, $$g(\alpha(x'')) \leq f(x'') +  2(\sqrt{2ns}+1) \eps.$$
(ii) For each lazily assigned point $y'' \in \mt'_g$, $$f(\beta(y'')) \leq g(y'') +  2(\sqrt{2ns}+1) \eps.$$
\end{lemma}
\begin{proof}
We only prove (i); (ii) is symmetric. The higher of the two roots of $\mt_f'$ and $\mt_g'$ is a matching point, and so are all the points at that height. Thus, all lazily assigned points have a matching point ancestor. We show that the nearest such ancestor cannot be too higher up.

Let $x$ be a matching point. We show that there exists a region $(\xi_1, \xi_2)$ as defined above with 
$$f(x) - 2(\sqrt{2ns}+1) \eps \leq \xi_2 \leq f(x).$$
Consider sweeping over the function values downward starting at $f(x)$ and let $\xi_2$ be the largest function value possible for a region as defined above. If the sweep line ever goes a distance greater than $2\eps$ without encountering a branching node or leaf in $\mt_f'$ or $\mt_g'$, then an $\xi_2$ is found. Therefore, there will be at least one branching node or leaf $x'$ in $\mt_f'$ or $\mt_g'$ per descent of $2\eps$ until $\xi_2$ is found. Suppose $\xi_2 < f(x) - 2(\sqrt{2ns} + 1)\eps$. Let $l = \sqrt{2ns} + 1$, and $f' = f(x) - 2l\eps$. Because each point in $\mt_f'$ and $\mt_g'$ has extent at least $2l\eps$, each branching node or leaf $x' \in \mt_f' \cup \mt_g'$ with height $f(x') \in [f', f(x))$ introduces at least one descendant in $\mt_f$ or $\mt_g$ at height $f'$. Since each edge length is at most $s\eps$, we can uniquely charge at least $(f(x') - f')/s\eps$ nodes to each branching node $x'$ at height $f(x') \in [f', f(x))$. For each leaf node $x'$, we can still charge $(f(x') - f')/s\eps$ nodes to $x'$; however these nodes can also be charged to at most one more branching node, namely the lowest branching node ancestor of $x'$ with height in $[f', f(x))$. Thus, each node can be charged at most twice. Since there is at least one branching node or leaf per descent of $2\eps$, the total number of nodes charged is at least
$$ \sum_{i=1}^l \frac{(l-i)2\eps}{s\eps} = \frac{l(l-1)}{s}.$$
Since each node is charged at most twice, we have 
$$\frac{l(l-1)}{s} \le 2n \Rightarrow l(l-1) \le 2ns,$$
a contradiction since $l = \sqrt{2ns} + 1$. Therefore, either $\xi_2 \ge f(x) - 2(\sqrt{2ns} + 1)\eps$, or the trees $\mt'_f$ and $\mt'_g$ do not extend below height $f(x) - 2(\sqrt{2ns}+1)\eps$. In either case, the lemma follows.
\end{proof}

\begin{lemma}
\label{lem:approx}
Let $\mt_f$ and $\mt_g$ be two merge trees and let $\eps > 0$ be a parameter. There is an $O(n^2)$ time
algorithm that returns a pair of $4(\sqrt{2ns}+1)\eps$-compatible maps between $\mt_f$ and $\mt_g$, if
$d_I(\mt_f, \mt_g) \leq \eps$ and the maximum length of a tree edge is $s\eps$. If
$d_I(\mt_f, \mt_g) > \eps$, then the algorithm may return no or return a pair of
$4(\sqrt{2ns}+1)\eps$-compatible maps.
\end{lemma}
\begin{proof}
Constructing the trees $\tilde{\mt}_f$ and $\tilde{\mt}_g$, the corresponding isomorphism between them (if it exists), and the maps $\alpha$ and $\beta$ between $\mt_f$ and $\mt_g$ (if they exist) takes time $O(n^2)$.

Except for the lazily assigned points, all the points in $\mt'_f$ and $\mt'_g$ are mapped by $\alpha$ and $\beta$ resp. to points at the same function value. By Lemma~\ref{lem:counting}, each point in $\mt_f'$ and $\mt_g'$ has its function value changed by at most $2(\sqrt{2ns}+1)\eps$. Points in $\mt_f - \mt'_f$ (resp. $\mt_g - \mt'_g$) have their nearest ancestors in $\mt'_f$ (resp. $\mt'_g$) at function value at most $2(\sqrt{2ns} + 1)\eps$ away. Since $\alpha$ and $\beta$ map them to the images of their nearest ancestors, their function values change by at most $2\cdot 2(\sqrt{2ns} + 1)\eps$.
\end{proof}

\noindent\emph{Remark.}
(i) Since the minimum edge length is $\le 2 \eps$, the maximum edge length is $s \eps$, and the ratio between the lengths of the longest and shortest edges is $r$; we have $r \ge s/2$.\\
(ii) If $s = \Omega(n)$, we modify the above algorithm slightly -- we skip the trimming step, but keep the rest same. It can be shown, as in Lemma \ref{lem:counting}, that the height of a point and its image differ by at most $2n\eps$.

\fakeparagraph{Putting it together.} By Lemmas~\ref{lem:easy_case} and \ref{lem:approx}, the decision procedure takes $O(n^{5/2})$ time. If it returns no, then $d_I(\mt_f, \mt_g) > \eps$. If it returns yes, then it also returns $O(\min\{n, \sqrt{rn}\}\eps)$-compatible maps between them. Hence, we conclude the following.
\begin{theorem}
\label{theorem:approx}
Given two merge trees $\mt_f$ and $\mt_g$ with a total of $n$ vertices, there exists an $O(n^{5/2} \log n)$ time $O\left(\min\{n, \sqrt{rn}\}\right)$-approximation algorithm for computing the interleaving distance between them, where $r$ is the ratio between the lengths of the longest and the shortest edge in both trees.
\end{theorem}

Combining Theorem \ref{theorem:approx} with Corollary \ref{cor:gh_i}, we have: 
\begin{corollary}
 Given two metric trees $\T_1$ and $\T_2$ with a total of $n$ vertices, there exists an $O(n^{7/2} \log n)$ time $O\left(\min\{n, \sqrt{rn}\}\right)$-approximation algorithm for computing the Gromov-Hausdorff distance between them, where $r$ is the ratio between the lengths of the longest and the shortest edge in both trees.
\end{corollary}


\section{Conclusion}
We have presented the first hardness results for computing the Gromov-Hausdorff distance between metric trees.
We have also given a polynomial time approximation algorithm for the problem. But the
current gap between the lower and upper bounds on the approximation factor is polynomially large. 
It would be very interesting to close this gap. 
In general, we hope that our current investigation will stimulate more research on the theoretical and algorithmic aspects of embedding or matching under additive metric distortion. 

\bibliographystyle{abbrv}
\bibliography{report_full}

\end{document}